\documentclass[12pt, twoside]{article}

\usepackage{amsmath,amsthm,amssymb}

\usepackage{times}

\usepackage{enumerate}

 \usepackage{amsmath}
\usepackage{amssymb}
\usepackage{amsthm}
\usepackage{graphicx}
\usepackage{amscd}
\usepackage{graphicx}

\usepackage{amssymb}
\usepackage{amsmath}
\usepackage{latexsym}
\usepackage[numbers,sort&compress]{natbib}
\usepackage{verbatim}
\usepackage{physics}
\usepackage{dsfont}	
\usepackage{microtype}  
\usepackage{hyperref}
\usepackage{tikz}
\usepackage{booktabs}   

\usepackage{color}
\newcommand{\one}[0]{\mathds{1}}

\newtheorem{theorem}{Theorem}
\newtheorem{lemma}[theorem]{Lemma}

\begin{document}

 \title{The geometry  of non-additive stabiliser codes}
\author{Simeon Ball and Pablo Puig\thanks{23 July 2021.}} 
 \date{}
\maketitle
 
 \begin{abstract}
We present a geometric framework for constructing additive and non-additive stabiliser codes which encompasses stabiliser codes and graphical non-additive stabiliser codes. 
\end{abstract}

\section{Introduction}
Error-correction is an essential component in the construction of a fault-tolerant quantum circuit \cite{AB1997}. The most prevalent class of quantum codes are stabiliser codes, introduced in \cite{Gottesman1997} and \cite{CRSS1998}. An $[\![n,k,d]\!]$ stabiliser code encodes $k$ logical qubits on $n$ physical qubits in such a way that there is a recovery map which is able to correct all errors of weight at most $\lfloor (d-1)/2 \rfloor$.  Here, an error of weight $w$ is a Pauli operator acting on $({\mathbb C}^2)^{\otimes n}$ which has precisely $n-w$ components which are the identity map. An $[\![n,k,d]\!]$ stabiliser code $Q(S)$ is described by an abelian subgroup $S$ of the Pauli group of size $2^{n-k}$. The code $Q(S)$ has dimension $2^k$ and is the intersection of the eigenspaces of eigenvalue $1$ of the linear operators of $S$. More generally, a $(\!(n,K,d)\!)$ is a code of dimension $K$ which encodes on $n$ physical qubits and for which there is a recovery map which is able to correct all errors of weight at most $\lfloor (d-1)/2 \rfloor$. Therefore, a $[\![n,k,d]\!]$ stabiliser code is a  $(\!(n,2^k,d)\!)$ code. 

It is well-established that there are parameters for which one can find direct sums of stabiliser codes which are larger than the optimal stabiliser code with the same $n$ and $d$. These codes are called {\em non-additive} stabiliser codes, as opposed to stabiliser codes which are often referred to as {\em additive stabiliser} codes, since they are equivalent to certain classical additive binary codes. For example, as a stabiliser code the optimal $[\![5,k,2]\!]$ code is attained by the $4$-dimensional $[\![5,2,2]\!]$ code. However, as discovered in \cite{RHSS1997}, there is a $(\!(5,6,2)\!)$ which is the direct sum of six $[\![5,0,3]\!]$ stabiliser codes. A simple description of this code was given using graphs in \cite{YCLO2007}, which also contained a construction of a $(\!(9, 12, 3)\!)$ non-additive stabiliser code. A subset of the same authors then provided an example of a $(\!(10,24,3)\!)$ code in \cite{YCO2007}.
Apart from the graphical non-additive stabiliser codes, there are also examples of direct sums of stabiliser codes constructed by Grassl and R\"otteler from Goethals and Preparata codes, see \cite{GR2008} and \cite{GR2013}. The latter article also gives a description of graphical non-additive stabiliser codes.

The aim of this article is to give a general geometrical framework for all these constructions. We start by giving an algebraic description of non-additive stabiliser codes, which are the direct sum of stabiliser codes. We then translate this construction to projective geometry and prove that such a code is given by a set of lines $\mathcal X$ with a specific property, called a {\em quantum set of lines}, and a set of points with the property that any pair of the points projects $\mathcal X$ onto a set of lines. 

The finite field with $q$ elements will be denoted ${\mathbb F}_q$. We will use the notation $[n,k]_q$ code to describe a linear $k$-dimensional code over of length $n$ over ${\mathbb F}_q$, i.e. a $k$-dimensional subspace of the vector space ${\mathbb F}_q^n$.

\section{Direct sum of stabiliser codes}
The following theorem is from Nielsen and Chuang \cite[Theorem 10.1]{NC2011} and is due to Bennett, DiVincenzo, Smolin and Wootters \cite{BDSW1996} and Knill and Laflamme \cite{KL1997}.
 
\begin{theorem} \label{qecthm}
Let $Q$ be a quantum code, let $P$ be the projector onto $Q$ and let $\mathcal E$ be a quantum operation.  A necessary and sufficient condition for the existence of an error-correction operation $R$ correcting $\mathcal E$ on $Q$ is that, for all $E_i, E_j \in \mathcal E$,
$$
PE_i^{\dagger}E_jP=\alpha_{ij}P,
$$
for some Hermitian matrix $\alpha$ of complex numbers.
\end{theorem}

Recall that the Pauli matrices are
$$
\one=\left(\begin{array}{cc} 1 & 0 \\ 0 & 1 \end{array}\right), \ X=\left(\begin{array}{cc} 0 & 1 \\ 1 & 0 \end{array}\right), \ Z=\left(\begin{array}{cc} 1 & 0 \\ 0 & -1 \end{array}\right), \ Y=\left(\begin{array}{cc} 0 & -i \\ i & 0 \end{array}\right).
$$

Let 
$$
\mathcal P_n=\{ c \sigma_1 \otimes \cdots \otimes \sigma_n \ | \ \sigma_i \in \{ \one, X, Y, Z \}, \ c^4=1 \}
$$ denote the group of Pauli operators on $({\mathbb C}^2)^{\otimes n}$. 

The {\em weight} $\mathrm{wt}(E)$ of $E \in \mathcal P_n$ is the number of non-identity operators in its tensor product.

We will be interested in constructing codes which can correct all errors in 
$$
\mathcal E_d=\{ E_i \in \mathcal P_n \ | \ \mathrm{wt}(E_i) \leqslant \lfloor (d-1)/2 \rfloor \}.
$$
By discretisation of errors, \cite[Theorem 10.2]{NC2011}, this allows such a code to correct any linear combination of the errors in $\mathcal E_d$.

Let $S$ be an abelian subgroup of $\mathcal P_n$ of size $2^{n-k}$. The additive stabiliser code $Q(S)$ is defined to be the intersection of the eigenspaces of eigenvalue $1$ of the elements of $S$. 
We will implicitly assume throughout that $S$ does not contain $-\one$ (so that $Q(S)$ is non-trivial).

\begin{theorem} \label{mindiststab}
Suppose $k\neq 0$. If $d$ is the minimum weight of $\mathrm{Centraliser}(S) \setminus S$ and we encode with $Q(S)$ then there is a recovery map which corrects all errors in $\mathcal E_d$.
\end{theorem}

\begin{proof}
Suppose $E_i,E_j \in \mathcal E$. Then $E=E_iE_j$ has weight at most $d-1$. This implies that 
$$
E\not\in \mathrm{Centraliser}(S) \setminus S
$$
since the elements of $\mathrm{Centraliser}(S) \setminus S$ have weight at least $d$. 

Thus, either $E \not\in \mathrm{Centraliser}(S)$ or $E \in S$.

The projector onto $Q(S)$ is
$$
P=\sum_{i=1}^{2^k} \ketbra{\psi_i},
$$
where $\{  \ket{\psi_i} \ | \ i=1,\ldots,2^k \}$ is an orthonormal basis for $Q(S)$.

If $E \not\in \mathrm{Centraliser}(S)$ then there is an element $M\in S$ such that $ME=-EM$ and
$$
PE_iE_jP=PEP=\sum_{r,s=1}^{2^k} \ketbra{\psi_r} E  \ketbra{\psi_s}
$$
$$
=\sum_{r,s=1}^{2^k}   \ketbra{\psi_r} E M\ketbra{\psi_s}
=-\sum_{r,s=1}^{2^k} \ketbra{\psi_r} ME \ketbra{\psi_s}=-PEP
$$
from which it follows that $PEP=0$.

If $E \in S$ then
$$
PE_iE_jP=PEP=\sum_{r,s=1}^{2^k} \ketbra{\psi_r} E  \ketbra{\psi_s}=\sum_{r,s=1}^{2^k}\ketbra{\psi_r}\!\!\!\ketbra{\psi_s}=P.
$$
Hence, Theorem~\ref{qecthm} implies there is a recovery map.
\end{proof}

In light of Theorem~\ref{mindiststab}, if $k\neq 0$ then one defines the minimum distance $d$ of $Q(S)$ to be the minimum weight of the elements of $\mathrm{Centraliser}(S) \setminus S$. If $k=0$ then we define the minimum distance $d$ of $Q(S)$ to be the minimum weight of the elements of $S$. 
If $d$ is the minimum weight of the elements of $\mathrm{Centraliser}(S)$ then the code is said to be {\em pure} and {\em impure} if not.

Suppose that $\{M_1,\ldots,M_{n-k}\}$ is a set of generators for $S$. We construct a binary $(n-k) \times 2n$ matrix, whose $j$-th row is obtained from the generator $M_j$ in the following way. If the $i$-th component of $M_j$ is $\one, X, Z, Y$ then the $(i,i+n)$ coordinates of the $j$-th row are $(0,0), (1,0), (0,1), (1,1)$ respectively. We denote this map by $\tau$, so the $j$-th row of $\mathrm{G}$ is $\tau(M_j)$. Let $C$ be the corresponding binary linear code with parameters $[2n,n-k]$ which has a generator matrix $G$. The fact that $S$ is abelian is equivalent to the property that for any two elements $u,v \in C$,
\begin{equation} \label{perp}
(u,v)=\sum_{i=1}^n (u_i v_{i+n}-v_i u_{i+n})=0. 
\end{equation}
This can be checked directly by observing that the only pairs of Pauli's that do not commute are $\{X,Y\}$, $\{X,Z\}$ and $\{Y,Z\}$ and that the only pairs $\{(u_i,u_{i+n}),(v_i, v_{i+n})\}$ that contribute a ``$1$'' to the sum are $\{(1,0),(1,1)\}$, $\{(1,0),(0,1)\}$ and $\{(1,1),(0,1)\}$.

If we define
$$
C^{\perp_s}=\{ v \in {\mathbb F}_2^{2n} \  | \ (u,v)=0,\ \mathrm{for} \ \mathrm{all} \ u\in C\}
$$
then the condition on $C$, so that $S$ is abelian, is that $C  \leqslant C^{\perp_s}$.

Let $T \subseteq {\mathbb F}_2^{n-k}$ and define, for $t=(t_1,\ldots,t_{n-k}) \in T$,
$$
Q_t(S)
$$
as the intersection of the eigenspaces of eigenvalue $1$ of $(-1)^{t_i}M_i$, for all $i \in \{1,\ldots,n-k\}$, and
$$
Q(S,T)=\bigoplus_{t \in T} Q_t(S).
$$

Let $t,u \in T \setminus \{0\}$ and let $A_{t,u}$ be a $(n-k) \times (n-k)$ non-singular matrix whose first two columns are $t$ and $u$. Then $A_{t,u}^{-1}G$ is also a generator matrix for $C$ and we can find another set 
$$
\{M_i' \ | i=1,\ldots,n-k \}
$$
of generators of $S$, where $M_i'$ is obtained from the $i$-th row of $A_{t,u}^{-1}G$ by applying $\tau^{-1}$, in other words reversing the construction above. 

We define $S_{t,u}$ as the subgroup of $S$ generated by $M'_3,\ldots,M'_{n-k}$. 

\begin{lemma} \label{stulemma}
Suppose $\ket{\psi^t} \in Q_t(S)$ and $\ket{\psi^{u}} \in Q_{u}(S)$. Then, for all $M\in S_{t,u}$, 
$$
M\ket{\psi^t}=\ket{\psi^t} \mathrm{and} \ M \ket{\psi^{u}}=\ket{\psi^{u}}.
$$
\end{lemma}

\begin{proof}
Observe that $Q_t(S)$ depends on the set of generators we have chosen for $S$. If we use the set of generators $M'_1,\ldots,M'_{n-k}$ for $S$ then $Q_t(S)$ becomes $Q_{(1,0,0,\ldots,0)}(S)$ and $Q_u(S)$ becomes $Q_{(0,1,0,\ldots,0)}(S)$. Thus, $M'_j \ket{\psi^t}=\ket{\psi^t}$ and $M'_j \ket{\psi^{u}}=\ket{\psi^{u}}$ for all $j\in \{3,\ldots,n-k\}$.
\end{proof}.

\begin{theorem} \label{mindistsum}
Let $T \subset {\mathbb F}_2^{n-k}$. If $d$ is the minimum weight of $\mathrm{Centraliser}(S_{t,u})$, where the minimum is taken over all pairs $(t,u)$ of non-zero elements of $T$, and we encode with $Q(S,T)$ then there is a recovery map which corrects all errors in $\mathcal E_d$.
\end{theorem}

\begin{proof}
The projector onto $Q(S,T)$ is
$$
P=\sum_{t \in T} \sum_{i=1}^{2^k} \ketbra{\psi^t_i}
$$
where $\{  \ket{\psi^t_i} \ | \ i=1,\ldots,2^k \}$ is an orthonormal basis for $Q_t(S)$.

Suppose $E_i,E_j \in \mathcal E$. Then $E=E_iE_j$ has weight at most $d-1$. This implies that 
$$
E \not\in \mathrm{Centraliser}(S_{t,u}) 
$$
for any $t,u \in T$, since the elements of $\mathrm{Centraliser}(S_{t,u})$ have weight at least $d$. 

Thus, since the elements in $\mathcal P_n$ either commute or anti-commute, there is an element $M_{t,u} \in S_{t,u}$ such that $M_{t,u}E =-E M_{t,u}$. 

By Lemma~\ref{stulemma},
$$
M_{t,u}\ket{\psi^t_r}=\ket{\psi^t_r} \mathrm{and} \ M_{t,u} \ket{\psi^{u}_s}=\ket{\psi^{u}_s},
$$
for all $r,s \in \{1,\ldots,2^k\}$.

Hence,
$$
PEP=\sum_{t,u \in T} \sum_{r,s=1}^{2^k} \ketbra{\psi^t_r} E  \ketbra{\psi^{u}_s}
$$
$$
=\sum_{t,u \in T} \sum_{r,s=1}^{2^k} \ketbra{\psi^t_r}E M_{t,u}\ketbra{\psi^{u}_s}
=-\sum_{t,u \in T} \sum_{r,s=1}^{2^k} \ketbra{\psi^t_r} M_{t,u}E \ketbra{\psi^{u}_s}=-PEP
$$
from which it follows that $PEP=0$. Theorem~\ref{qecthm} implies there exists a recovery map.
\end{proof}

In light of Theorem~\ref{mindistsum}, we conclude that the minimum distance of $Q(S,T)$ is at least the minimum of the minimum weight of the elements of $\mathrm{Centraliser}(S_{t,u})$ as $t$ and $u$ run over all pairs of non-zero elements of $T$.

Note the difference between Theorem~\ref{mindiststab} and Theorem~\ref{mindistsum}. In the latter case there are no errors which act trivially on the code space. This is due to the fact that for distinct $u,t \in T$ there is a $j$ for which $u_j \neq t_j$ and for this $j$, $M_j\ket{\psi^u} \neq M_j\ket{\psi^t}$.

\section{The geometry of a direct sum of stabiliser codes}

Let PG$(k-1,2)$ denote the $(k-1)$-dimensional projective geometry over ${\mathbb F}_2$,  the field of two elements. This geometry consists of points which are the non-zero vectors of ${\mathbb F}_2^k$ and lines, which are three points $\{u,v,u+v\}$, and higher dimensional subspaces, which are obtained from subspaces of ${\mathbb F}_2^k$ by removing the zero vector.

If $\mathrm{Centraliser}(S)$ does not have any elements of weight one then an additive $[\![n,k,d]\!]$ stabiliser code $Q(S)$ is entirely equivalent to a set $\mathcal X$ of $n$ lines in PG$(n-k-1,2)$ with the property that any co-dimension two subspace is skew to (does not intersect) an even number of the lines in $\mathcal X$, see \cite{GGMG} and \cite[Lemma 3.6]{BCH2021}. 

A set of generators of the abelian subgroup $S$ can be obtained from $\mathcal X$ by constructing a $(n-k) \times 2n$ matrix $\mathrm{G}$, whose $i$-th and $(i+n)$-th column is a basis for the $i$-th line of $\mathcal X$, where $i \in \{1,\ldots,n\}$.  
Recall from the previous section that we defined a map $\tau$  so that the $j$-th row of $\mathrm{G}$ is $\tau(M_j)$, where $M_1,\ldots,M_{n-k}$ is a set of generators for $S$. The code generated by $\mathrm{G}$ will be denoted by $C=\tau(S)$. If $Q(S)$ is pure then the minimum distance $d$ can be obtained from the geometry as the size of the minimum set of dependent points on distinct lines of $\mathcal X$, see \cite{GGMG} or \cite{BCH2021}. For the sake of completeness, observe that $C^{\perp_s}=\tau(\mathrm{Centraliser}(S))$. The symplectic weight of an element $v \in {\mathbb F}_2^{2n}$ is the size of the support
$$
\mathrm{Support}(v)= \{ i \in \{1,\ldots,n \} \ | \ (v_i,v_{i+n}) \neq (0,0)\}.
$$

Since an element of $v\in C^{\perp}_s$ is symplectically orthogonal to all the rows of $\mathrm{G}$, an element of symplectic weight $w$ will give a dependence of $w$ points on the $w$ lines of $\mathcal X$ corresponding to the elements of $\mathrm{Support}(v)$.

In the case of impure codes we have to discount the dependencies in which the lines of $\mathcal X$ which do not contain dependent points are contained in a hyperplane (a co-dimension one subspace), which also contains the dependent points, see \cite{BCH2021}. However, for the purposes of this article, Theorem~\ref{mindistsum} bounds the minimum distance of $Q(S,T)$ below by the minimum of $\mathrm{Centraliser}(S_{t,u})$. This is obtained geometrically as the size of the minimum set of dependent points on distinct lines of $\mathcal X_{t,u}$, obtained from the subgroup $S_{t,u}$.

Let $T$ be a subset of ${\mathbb F}_2^{n-k}$. Let $t,u$ be distinct non-zero elements of $T$ and let $A_{t,u}$ be a $(n-k) \times (n-k)$ non-singular matrix whose first two columns are $t$ and $u$. In the previous section we noted that $A_{t,u}^{-1}G$ is another generator matrix for $C$ and that we can find another set 
$$
\{M_i' \ | i=1,\ldots,n-k \}
$$
of generators of $S$, where $M_i'$ is obtained from the $i$-th row of $A_{t,u}^{-1}G$. We then defined $S_{t,u}$ as the subgroup of $S$ generated by $M'_3,\ldots,M'_{n-k}$. As above, let $\mathcal X_{t,u}$ be the quantum set of lines of PG$(n-k-3,2)$ we obtain from the subgroup $S_{t,u}$. The geometric path from $\mathcal X$ to $\mathcal X_{s,t}$ is to project the set of lines $\mathcal X$ to a set of lines in PG$(n-k-3,2)$ from the points $t$ and $u$. Recall that to project from the $i$-th canonical basis element we simply delete the $i$-th coordinate. Therefore, after changing the basis with $A_{t,u}$, we project from $t$ and $u$ by deleting the first two coordinates. In the subgroup setting this is equivalent to removing $M_1'$ and $M_2'$ from the set of generators.

Thus, the code $Q(S,T)$ is equivalent to a set $\mathcal X$ of $n$ lines in PG$(n-k-1,2)$ with the property that any co-dimension two subspace is skew to an even number of the lines of $\mathcal X$, together with a set of points $T \setminus \{0\}$ whose pairs project the lines of $\mathcal X$ onto a set of lines in PG$(n-k-3,2)$. Since every co-dimension two subspace of PG$(n-k-1,2)$ is skew to an even number of the lines of $\mathcal X$, it is immediate that in the projection this property holds too. Thus, the projection from $t$ and $u$ of $\mathcal X$ is onto a quantum set of lines $\mathcal X_{t,u}$ in PG$(n-k-3,2)$ which gives the subgroup $S_{t,u}$, by the construction described above. Therefore $d(\mathcal X_{t,u})$, the size of the minimum set of dependent points on distinct lines of $\mathcal X_{t,u}$, as $t$ and $u$ vary over all pairs of non-zero elements of $T$, is a lower bound for the minimum distance of $Q(S,T)$. We have proved the following theorem.

\begin{theorem} \label{qstthm}
Let $T \subset {\mathbb F}_2^{n-k}$. Let $\mathcal X$ be the quantum set of lines given by the abelian subgroup $S$. The code $Q(S,T)$ is a $(\!(n,|T|2^k,d)\!)$ code, where 
$$
d \geqslant \min_{t,u} d(\mathcal X_{t,u}),
$$
as $t$ and $u$ vary over all pairs of non-zero elements of $T$.
\end{theorem}

There are at least two possible ways to proceed to use Theorem~\ref{qstthm}. 

The most straightforward would be to start with an abelian subgroup $S$, where $Q(S)$ is a pure $[\![n,2^k,d']\!]$ code.  This will allow us to construct a quantum set of lines $\mathcal X$ in PG$(n-k-1,2)$ and try to find the largest $T \subset {\mathbb F}_2^{n-k}$ with the property that $d(\mathcal X_{t,u})$ is at least $d$, as $t$ and $u$ vary over all pairs of non-zero elements of $T$. One may choose any $d \leqslant d'$, although in many cases that choosing $d=d'$ results in $T=\{0\}$ and one is not able to construct anything more than $Q(S)$.

We construct a graph $\Gamma$ whose vertices are the points of $T$ and where $t$ and $u$ are joined by an edge if and only if $d(\mathcal X_{t,u}) \geqslant d$ and choose $T$ to be a largest clique in $\Gamma$.

Consider for example the $[\![5,0,3]\!]$ code $Q(S)$, where $S$ is the abelian subgroup generated by
$$
\begin{array}{rcl}
M_1 & = & XZ\one \one Z \\
M_2 & = & ZXZ\one \one  \\
M_3 & = & \one ZXZ\one  \\
M_4 & = & \one \one ZXZ \\
M_5 & = & Z\one \one Z X\\
\end{array}
$$
Here, we are suppressing the tensor product symbol between the matrices.

Following the discussion above, the matrix
$$
\mathrm{G}=
\left(\begin{array}{ccccc|ccccc} 
1 & 0 & 0 & 0 & 0 & 0 & 1 & 0 & 0 & 1 \\ 
0 & 1 & 0 & 0 & 0 & 1 & 0 & 1 & 0 & 0 \\ 
0 & 0 & 1 & 0 & 0 & 0 & 1 & 0 & 1 & 0 \\ 
0 & 0 & 0 & 1 & 0 & 0 & 0 & 1 & 0 & 1 \\ 
0 & 0 & 0 & 0 & 1 & 1 & 0 & 0 & 1 & 0 \\ 
\end{array}\right)
$$
and the set of lines
$$
\mathcal X=\{ 
\langle e_1,e_2+e_5 \rangle, 
\langle e_2,e_1+e_3 \rangle, 
\langle e_3,e_2+e_4 \rangle,
\langle e_4,e_3+e_5 \rangle,
\langle e_5,e_1+e_4 \rangle
\},
$$
where $e_i$ is the $i$-th element in the canonical basis of ${\mathbb F}_2^5$.

There are $16$ points in PG$(4,2)$ which are not incident with any line of $\mathcal X$. We define a graph $\Gamma$ whose vertices are these 16 points and where two points $t$ and $u$ are joined by an edge if and only if they project $\mathcal X$ onto a set of lines. Note that the condition $d(\mathcal X_{t,u}) \geqslant 2$ is redundant. The edge condition can be verified by checking that $t,u$ and $x$ are linearly independent for any point $x$ incident with a line of $\mathcal X$.

A short computation using GAP \cite{GAP2020} reveals that $\Gamma$ has $60$ edges and $6$ cliques of size $5$. Thus, choosing one of these, we set
$$
T=\{ (0,0,0,0,0) ,e_1+e_2+e_4,e_2+e_3+e_5,e_1+e_3+e_4,e_2+e_4+e_5,e_1+e_3+e_5\}.
$$
Then, Theorem~\ref{qstthm} implies $Q(S,T)$ is a $(\!(5,6,2)\!)$ code.

The second possible way to apply Theorem~\ref{qstthm} is to fix $T$ and then try and construct $\mathcal X$ (and hence $S$). 
Suppose, as in the previous paragraph, we would like to construct a $(\!(5,6,2)\!)$ code. We need to find a quantum set of lines $\mathcal X$, with the property that no point incident with a line of $\mathcal X$ is spanned by two points of $T$. This will ensure that the projection of $\mathcal X$ from any two points of $T$ is onto a set of lines of PG$(2,2)$.

If four of the elements of $T\setminus \{0\}$ span a two-dimensional subspace $\pi$, (i.e. a PG$(2,2)$) then the lines of $\mathcal X$ must be skew to $\pi$, otherwise there is a point incident with a line of $\mathcal X$ which is in the span of two points of $T\setminus \{0\}$. This contradicts the fact that $\mathcal X$ is a quantum set of lines. Likewise, if five of the elements of $T\setminus \{0\}$ span a three-dimensional subspace $\pi$ then any point of $\pi$ is in the span of two points of $T$, which implies that the lines  of $\mathcal X$ must be skew to $\pi$, a hyperplane of PG$(4,2)$, which is impossible.

Thus, we can assume the elements of $T\setminus \{0\}$ are linearly independent and can choose a basis so that 
$$
T=\{ (0,0,0,0,0) ,e_1,e_2,,e_3,e_4,e_5\}.
$$
We can now try and deduce $S$. The projection $\mathcal X_t$ of $\mathcal X$, from any point of $t \in T \setminus \{0\}$, should be a set of $5$ lines in $\pi_t$, a three-dimensional space PG$(3,2)$. These lines are not incident with the basis points, otherwise the projection onto PG$(2,2)$ from two points of $T$ would not be a set of lines. Since $\mathcal  X_t$ is a quantum set of lines  it has the property that every line of $\pi_t$ is skew to an even number of the lines of $\mathcal X_t$. Since no line of $\mathcal X_t$ is incident with a basis point, each weight $2$ point (a point spanned by two basis points) must be incident with a line of $\mathcal X_t$. Furthermore, every line not incident with a basis point is incident with a weight two point. Therefore, four of the lines of $\mathcal X_t$ are incident with one weight two point and one of them is incident with two. Up to permutation of coordinates suppose the latter line joins $e_1+e_4$ and $e_2+e_3$. The other four lines consist of two weight 3 points and a weight 2 point. Therefore, up to a permutation of the coordinates, the unique configuration of lines is given in Figure~\ref{qsolxt}. The lines of $\mathcal X_t$ are in bold.

\begin{figure}[h] 
\begin{center}
\begin{tikzpicture} [scale=1.4]
shorten > = .1pt, auto,thick]
\filldraw[black] (0,0) circle (1.5pt) node[label=left:$e_1+e_3$] (v0) {};
\filldraw[black] (2,3.46) circle (1.5pt) node[label=above:$e_2+e_3$]  (v1) {};
\filldraw[black] (4,0) circle (1.5pt) node[label=right:$e_3+e_4$] (v2) {};
\filldraw[black] (1,1.75) circle (1.5pt) node[label=left:$e_1+e_2$] (v3) {};
\filldraw[black] (3,1.75) circle (1.5pt) node[label=right:$e_2+e_4$] (v4) {};
\filldraw[black] (2,-0.6) circle (1.5pt) node[label=below:$e_1+e_3+e_4$] (v6) {};
\filldraw[black] (0,3.46) circle (1.5pt) node[label=left:$e_1+e_4$] (v7) {};
\filldraw[black] (4,3.46) circle (1.5pt) node[label=right:$e_1+e_2+e_3+e_4$] (v8) {};
   \coordinate (c1) at (intersection of v0--v4 and v2--v3);

        \coordinate (c2) at (intersection of v3--v6 and v0--c1);

             \coordinate (c3) at (intersection of v4--v6 and v2--c1);
\node[label=above:$e_1+e_2+e_4$] at (c1) {};
\node[label=right:$e_1+e_2+e_3$] at (c3) {};
\node[label=left:$e_2+e_3+e_4$] at (c2) {};
\path[lightgray] (v0) edge node {} (v1);
\path[lightgray] (v2) edge node {} (v1);
\path[lightgray] (v6) edge node {} (c1);
\path[lightgray] (v1) edge node {} (c1);
\path[black] (v0) edge node {} (c1);
\path[black] (v7) edge node {} (v8);
\path[black] (v2) edge node {} (c1);
\path[black] (v3) edge node {} (v6);
\path[black] (v4) edge node {} (v6);
     \fill[black] (c2) circle (1.5pt) ;
          \fill[black] (c1) circle (1.5pt);
     \fill[black] (c3) circle (1.5pt);
     
\end{tikzpicture}
\caption{The unique quantum set of lines $X_t$ not incident with the basis points.}
\label{qsolxt}
\end{center}
\end{figure}
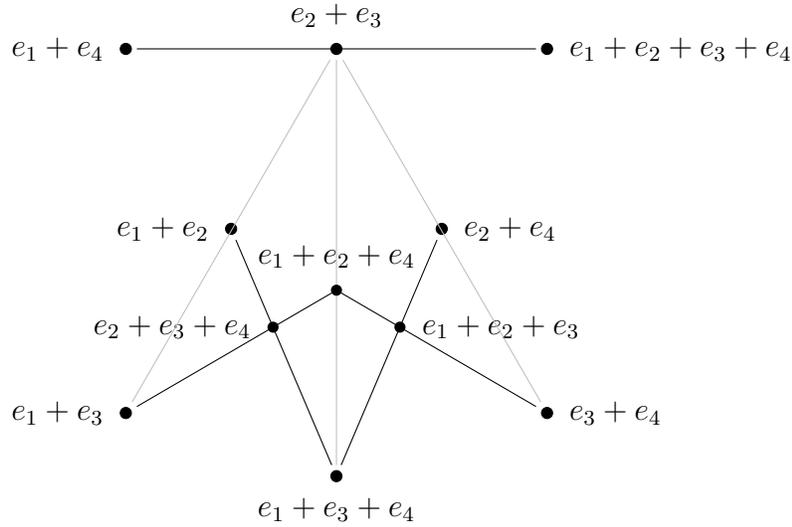

Therefore, up to permutation of the coordinates, we deduce that four of the five rows of $\mathrm{G}$ are
$$
\left(\begin{array}{ccccc|ccccc} 
0 & 0 & 1 & 1 & 1 & 1 & 1 & 1 & 0 & 1 \\ 
1 & 0 & 0 & 1 & 1 & 1 & 1 & 1 & 1 & 0 \\ 
1 & 1 & 0 & 0 & 1 & 0 & 1 & 1 & 1 & 1 \\ 
1 & 1 & 1 & 0 & 0 & 1 & 0 & 1 & 1 & 1 \\ 
\end{array}\right)
$$

Since the projection of any two of the basis points projects onto a point of PG$(2,2)$ there can be no points of weight two on the lines of $\mathcal X$. Therefore
$$
\mathrm{G}=
\left(\begin{array}{ccccc|ccccc} 
u_1 & 1 & 1 & 1 & u_5 & u_1+1 & u_7 & 0 & u_9 & u_{5}+1 \\
0 & 0 & 1 & 1 & 1 & 1 & 1 & 1 & 0 & 1 \\ 
1 & 0 & 0 & 1 & 1 & 1 & 1 & 1 & 1 & 0 \\ 
1 & 1 & 0 & 0 & 1 & 0 & 1 & 1 & 1 & 1 \\ 
1 & 1 & 1 & 0 & 0 & 1 & 0 & 1 & 1 & 1 \\ 
\end{array}\right)
$$
for some $u_1,u_5,u_7,u_9$. Since  
$$
\sum_{i=1}^n (u_i v_{i+n}-v_i u_{i+n})=0,
$$
for any two rows $u$ and $v$ of $\mathrm{G}$, we deduce that
$$
u_1=u_5\neq u_7=u_9.
$$
This gives two solutions which generate the same subgroup, the subgroup $S'$ generated by
$$
\begin{array}{rcl}
M_1' & = & ZYXYZ \\
M_2'& = & ZZYXY \\
M_3' & = &YZZYX \\
M_4' & = & XYZZY \\
M_5' & = & YXYZZ\\
\end{array}
$$
Observe that $M_i'M_{i+1}'M_{i+3}'=M_i$ (indices read modulo $n$), so $S'=S$.

Thus we have proved that, up to permutation of the non-identity Pauli operators in a coordinate and a permutation of the coordinates (the qubits), the $(\!(5,6,2)\!)$ code is unique.

\section{Stabiliser codes as direct sums of stabiliser codes}

In this section we investigate the problem of determining when $Q(S,T)$ is itself a stabiliser code. Obviously a necessary condition is that $|T|=2^r$ for some $r$. In the following theorem, we prove a sufficient condition.

\begin{theorem} \label{stabasdirectsum}
Let $S$ be an abelian group of size $2^{n-k}$ and let $T$ be an $r$-dimensional subspace. Then $Q(S,T)=Q(S')$ for some subgroup $S'$ of $S$ of size $2^{n-r-k}$.
\end{theorem}

\begin{proof}
By applying a change of basis, we can assume that
$$
T=(0,\ldots,0) \cup \langle e_{n-k-r+1},\ldots,e_{n-k} \rangle.
$$

Let $\{M_1,\ldots,M_{n-k}\}$ be a set of generators of $S$. For all $\ket{\psi} \in Q(S,T)$,
$$
M_i \ket{\psi}=\ket{\psi},
$$
for $i \in \{1,\ldots,n-k-r\}$. Hence, 
$$
Q(S,T) \leqslant Q(S'),
$$
where the subgroup $S'$ is generated generated by $\{M_1,\ldots,M_{n-r-k}\}$.

Since $\dim Q(S,T)=\dim Q(S')=2^{r+k}$, we have $Q(S,T)=Q(S')$.
\end{proof}

Is is tempting to believe that the contrary statement is also true. That if $Q(S,T)=Q(S')$ for some subgroup $S'$ of $S$ then $T$ must be a subspace. However, this is not the case. For example, if
$$
T=\{e_1,e_2\}
$$
then $\dim Q(S,T)=2^{k+1}$ and since 
$$
Q(S,T) \leqslant Q(S'),
$$
where $S'$ is generated by $-M_1M_2,M_3,\ldots,M_{n-k}$, we conclude that $Q(S,T)=Q(S')$.

The following theorem, which is equivalent to \cite[Theorem 2]{YCO2007}, states that any stabiliser code can be obtained as a direct sum of one-dimensional stabiliser codes.

\begin{theorem}
Let $Q(S')$ be a $[\![n,k,d]\!]$ stabiliser code. Then $Q(S')=Q(S,T)$ for some $S \supseteq S'$ of size $2^n$ and some $k$-dimensional subspace $T \subset {\mathbb F}_2^{n}$.  Hence, any stabiliser code is the direct sum of $[\![n,0,d']\!]$ stabiliser codes for some $d' \geqslant d$.
\end{theorem}

\begin{proof}
Let $\{M_1,\ldots,M_{n-k}\}$ generate $S'$. We can extend $S'$ to an abelian subgroup $S$ of size $2^n$, where $S' \supseteq S$. This is most easily seen in the binary code setting, where we can extend $\tau(S)=C < C^{\perp_s}$, to a code $C'>C$ such that $C'=(C')^{\perp}$. We can extend$\{M_1,\ldots,M_{n-k}\}$ to a set $\{M_1,\ldots,M_{n}\}$ which generate $S'$. If we then set
$$
T=\langle e_{n-k+1},\ldots,e_n\rangle,
$$
we have that $Q(S')=Q(S,T)$.

Note that since $\mathrm{Centraliser}(S') \geqslant \mathrm{Centraliser}(S)$, it follows that $d' \geqslant d$.
\end{proof}

\section{Graphical non-additive stabiliser codes} \label{secgraph}

The case $k=0$ is equivalent to graphical quantum error-correcting codes. To see this, note that we can choose a basis for the geometry so that the initial $n\times n$ matrix of $\mathrm{G}$ is the identity matrix. We can then choose a basis for each line of $\mathcal X$ so that the $i$-th coordinate of the $(i+n)$-th column in zero. The matrix $\mathrm{G}$ is then of the form $(\mathrm{I}_n \ | \ \mathrm{A})$ for some $n \times n$ matrix $\mathrm{A}$. The condition (\ref{perp}) implies that $\mathrm{A}$ is symmetric, so we can interpret $\mathrm{A}$ as the adjacency matrix of a simple graph $\Gamma$ on $n$ vertices. The elements of $T$ can then be described by colouring the appropriate vertices in $|T|$ copies of the graph, see \cite[Figure 1]{YCO2007}. 

In \cite{YCO2007}, the set $T$ is called a {\em coding clique}. The condition in Theorem~\ref{qstthm} is given as a purely combinatorial condition. One makes a set $R$ of subsets of $\{1,\ldots,n\}$ which consists of, for each subset $U$ of the vertices of $\Gamma$ of size at most $d-1$, the symmetric difference of the neighbourhood of $U$. One then deduces the largest set $T$ of subsets of $\{1,\ldots,n\}$ with the property that the symmetric difference of any two elements of $T$ is not an element of $R$.

Theorem~\ref{qstthm} allows us to interpret this condition geometrically. We consider $U$ as a subset of at most $d-1$ points incident with distinct lines of $\mathcal X$. We let $R$ be the set of points of PG$(n-1,2)$ which are in the span of the points in $U$. The set $T$ is a set of points of PG$(n-1,2)$ with the property that no two points of $T$ span a point in $R$.

Let us consider, as an example, the $(\!(9,12,3)\!)$ code. The matrix
      
$$
 \mathrm{G}=
\left(\begin{array}{ccccccccc|ccccccccc} 
    1& 0& 0& 0& 0& 0& 0& 0& 0& 0& 1& 0& 0& 0& 0& 0& 0& 1\\
      0& 1& 0& 0& 0& 0& 0& 0& 0& 1& 0& 1& 0& 0& 0& 0& 0& 0\\
      0& 0& 1& 0& 0& 0& 0& 0& 0& 0& 1& 0& 1& 0& 0& 0& 0& 0\\
      0& 0& 0& 1& 0& 0& 0& 0& 0& 0& 0& 1& 0& 1& 0& 0& 0& 0\\
      0& 0& 0& 0& 1& 0& 0& 0& 0& 0& 0& 0& 1& 0& 1& 0& 0& 0\\
      0& 0& 0& 0& 0& 1& 0& 0& 0& 0& 0& 0& 0& 1& 0& 1& 0& 0\\
      0& 0& 0& 0& 0& 0& 1& 0& 0& 0& 0& 0& 0& 0& 1& 0& 1& 0\\
      0& 0& 0& 0& 0& 0& 0& 1& 0& 0& 0& 0& 0& 0& 0& 1& 0& 1\\
      0& 0& 0& 0& 0& 0& 0& 0& 1& 1& 0& 0& 0& 0& 0& 0& 1& 0
          \end{array}
      \right)
$$

and the set of lines
$$
\mathcal X=\{ 
\langle e_1,e_2+e_9 \rangle, 
\langle e_2,e_1+e_3 \rangle, 
\langle e_3,e_2+e_4 \rangle,
\langle e_4,e_3+e_5 \rangle,
\langle e_5,e_4+e_6\rangle,
$$
$$
\langle e_6,e_5+e_7 \rangle, 
\langle e_7,e_6+e_8 \rangle, 
\langle e_8,e_7+e_9 \rangle,
\langle e_9,e_1+e_8 \rangle
\}.
$$

We consider the span of two points on lines of $\mathcal X$ and their intersection with the $5$-dimensional subspace $\pi$, defined by
$$
X_2+X_6=0, \ X_3+X_8=0, \ X_5+X_9=0.
$$
One can quickly verify that only 27 points of $\pi$ are in the span of two points incident with lines of $\mathcal X$. We restrict the vertices of the graph $\Gamma$ to the remaining 36 points of $\pi$. A quick calculation on GAP shows that this graph has 12 cliques of size $11$. The structure of the $11$ non-zero points of $T$, obtained from one of these cliques, is a cone with vertex point $(1,0,0,1,0,0,1,0,0)$ and a base of five linearly independent points.
For example, one can take $T$ be be the following vectors.
$$
\begin{array}{cccc}
(  0,  0, 0,  0, 0,  0,  0, 0, 0 ) &
 (  0,  0, 1,  0, 1,  0,  0, 1, 1 ) & 
  (  0,  0, 1, 1,  0,  0,  0, 1,  0 ) & 
  (  0, 1,  0,  0, 1, 1,  0,  0, 1 )\\
  (  0, 1,  0, 1,  0, 1, 1,  0,  0 )& 
  (  0, 1, 1,  0, 1, 1, 1, 1, 1 ) &
  ( 1,  0,  0, 1,  0,  0, 1,  0,  0 ) &
  ( 1,  0, 1,  0,  0,  0, 1, 1,  0 ) \\
  ( 1,  0, 1, 1, 1,  0, 1, 1, 1 )&
  ( 1, 1,  0,  0,  0, 1,  0,  0,  0 )&
  ( 1, 1,  0, 1, 1, 1, 1,  0, 1 )&
  ( 1, 1, 1, 1, 1, 1,  0, 1, 1 ) 
  \end{array}
  $$

\section{Qupit non-additive stabiliser codes}

Perhaps the most useful aspect of the geometrical construction of non-additive stabiliser codes is that it directly generalises to the qupit case, i.e. when the local dimension is any prime $p$. There are a few differences that need to be pointed out. The points of PG$(n-k-1,p)$ are the one-dimensional subspaces of ${\mathbb F}_p^{n-k}$ and lines are two-dimensional subspaces of ${\mathbb F}_p^{n-k}$. Note that there are $p+1$ one-dimensional subspaces contained in a two-dimensional subspace, so in the geometry there are $p+1$ points incident with a line. The condition that if $\mathcal X$ is a quantum set of lines of PG$(n-k-1,p)$ then every co-dimension two subspace is skew to an even number of lines no longer holds. However, given an abelian subgroup $S$, the construction of the quantum set of lines $\mathcal X$ follows in the same way. Following Ketkar et al \cite{KKKS2006}, we define the Pauli operators on $({\mathbb C}^p)^{\otimes n}$ as follows.

Let  $\{ \ket x \ | \ x \in {\mathbb F}_p^n\}$ be a basis of $({\mathbb C}^p)^{\otimes n}$ and let $\omega$ be a primitive complex $p$-th root of unity. Define
$$
X(a)\ket x=\ket{x+a}
$$
for each $a \in {\mathbb F}_p^n$ and
$$
Z(b)\ket x=\omega^{x \cdot  b}\ket x.
$$
for each $b \in {\mathbb F}_p^n$.

The Pauli group, for $p \geqslant 3$, is
$$
\{ \omega^cX(a)Z(b) \ | \ a,b \in {\mathbb F}_p^n, c \in  {\mathbb F}_p\}.
$$

We define the non-additive stabiliser code for a subset $T \subseteq {\mathbb F}_p^{n-k}$ and an abelian subgroup $S$ of the Pauli  group as before. For $t \in T$,
$$
Q_t(S)
$$
is the intersection of the eigenspaces of eigenvalue $1$ of $\omega^{t_i}M_i$ ($i =1,\ldots,n-k$) and
$$
Q(S,T)=\bigoplus_{t \in T} Q_t(S).
$$
For $t,u\in T\setminus \{0\}$ defining distinct points of PG$(n-k-1,p)$, the set of lines $\mathcal X_{t,u}$ is again defined as the set of lines of PG$(n-k-3,p)$ obtained from $\mathcal X$ be projection from $t$ and $u$. 

Then all proofs work as before, although in the proof of Theorem~\ref{mindistsum} we need to modify slightly the argument. If $E \not\in \mathrm{Centraliser}(S_{t.u})$ then we deduce that there is an $M_{t,u} \in S_{t,u}$ such that
$$
EM_{t,u}=\omega^i M_{t,u}E,
$$
for some $i \in \{1,\ldots,p-1\}$. Note that, since 
$$
EM_{t,u}^j=\omega^{ij} M_{t,u}^jE,
$$
we can always find an $M_{t,u} \in S_{t,u}$ such that
$$
EM_{t,u}=\omega M_{t,u}E.
$$
Thus, we have that 
$$
PEP=\omega PEP,
$$
from which it follows that $PEP=0$.

Theorem~\ref{qstthm} generalises to the following theorem, where the subscript in $(\!(n,K,d)\!)_p$ indicates the local dimension.

\begin{theorem} \label{allpthm}
Let $T \subset {\mathbb F}_p^{n-k}$. Let $\mathcal X$ be the quantum set of lines given by the abelian subgroup $S$. The code $Q(S,T)$ is a $(\!(n,|T|p^k,d)\!)_p$ code, where $d$ is at least the size of the minimum set of dependent points on distinct lines of $\mathcal X_{t,u}$, as $t$ and $u$ vary over all pairs of non-zero elements of $T$ defining distinct points of PG$(n-k-1,p)$.
\end{theorem}

 For example, let $\mathcal X$ be the quantum set of $11$ lines PG$(6,3)$ obtained from the following $7 \times 22$ matrix over ${\mathbb F}_3$,
$$
\mathrm{G}=
\left(\begin{array}{ccccccccccc|ccccccccccc}
  1& 2& 2& 0& 0& 2& 0& 2& 0& 2& 
      2& 0& 1& 2& 2& 2& 0& 0& 0& 1& 
      2& 0 \\
   0& 1& 0& 2& 2& 0& 1& 1& 1& 
      0& 0& 2& 0& 0& 2& 2& 0& 2& 0& 
      1& 1& 0 \\ 
   1& 0& 2& 1& 1& 0& 2& 0& 0& 2& 
      0& 0& 1& 1& 0& 1& 0& 0& 1& 
      2& 0& 0 \\ 
   0& 1& 2& 2& 1& 1& 2& 2& 0& 0& 
      1& 0& 2& 0& 0& 0& 0& 2& 1& 
      2& 1& 0 \\ 
   2& 2& 1& 1& 0& 0& 2& 2& 0& 2& 
      2& 0& 2& 0& 1& 1& 2& 1& 1& 
      1& 2& 1 \\ 
   0& 0& 1& 2& 0& 2& 2& 2& 0& 0& 
      2& 2& 1& 2& 1& 1& 2& 0& 1& 2& 
      2& 0 \\ 
   0& 0& 0& 1& 1& 1& 2& 0& 1& 
      0& 2& 2& 0& 1& 1& 2& 0& 1& 2& 
      0& 0& 0 
      \end{array}
      \right)
  $$
and let $T$ be the set of the following nine points.
$$
\begin{array}{ccc}
  ( 0, 0, 0, 0, 0, 0, 0 ) &   ( 1, 0, 0, 0, 0, 0, 1 ) & ( 2, 0, 0, 0, 0, 0, 2 )\\
  ( 1, 0, 1, 1, 0, 1, 1 ) & ( 2, 0, 2, 2, 0, 2, 2 ) &  (1, 0, 2, 2, 0, 2, 1 ) \\
   (2, 0, 1, 1, 0, 1, 2 ) & ( 0, 0, 1, 1, 0, 1, 0 ) &  ( 0, 0, 2, 2, 0, 2, 0 )\\
   \end{array}
  $$
One can check that the projection from any pair of non-zero points $t,u$ of $T$ is onto a quantum set of lines $\mathcal X_{t,u}$ of PG$(4,3)$ with the property that no point is incident with two lines of $\mathcal X_{t,u}$. This latter property implies that the size of the minimum set of dependent points on distinct lines of $\mathcal X_{t,u}$ is at least $3$, since two points are dependent if and only if they are the same point. One can check that $T$ is a subspace so, by Theorem~\ref{stabasdirectsum}, $Q(S,T)$ is a $[\![11,6,3]\!]_3$ stabiliser code. Furthermore, this is an optimal stabliser code for an $[\![11,k,3]\!]_3$ code since there is no quantum MDS code (a code attaining the quantum Singleton bound) with these parameters. Recall that the quantum Singleton bound, proved by Rains in \cite{Rains1999}, states that
$$
k\leqslant n-2(d-1),
$$
which in this case gives $k \leqslant 7$. However, the existence of an $[\![11,7,3]\!]_3$ stabiliser code can be ruled out, since there is no additive MDS code of length $11$ over ${\mathbb F}_9$, see \cite{BGL2021}.

 \section{A recipe for constructing non-additive stabiliser codes}
 
Theorem~\ref{allpthm} leads to the following recipe for the construction of non-additive stabiliser codes of length $n$ and minimum distance $d$.

\begin{itemize}
\item
Choose a graph on $n$ vertices whose edges are labelled by elements of ${\mathbb F}_p$ (non-edges are labelled by zero) and form the $n \times 2n$ matrix $\mathrm{G}=(\mathrm{I}_n \ | \ \mathrm{A})$, where $\mathrm{A}$ is the (symmetric) adjacency matrix of the graph. 
\item
Let $\mathcal X$ be  the quantum set of $n$ lines of PG$(n-1,p)$, whose $i$-th line is the span of the $i$-th and $(i+n)$-th column of $\mathrm G$ and let $P$ be the set of points which are incident with a line of $\mathcal X$.
\item
Either calculate the set $R$ of points of PG$(n-1,p)$ which are not in the span of $d-1$ or less points of $P$ and choose $k$ linearly independent points $K$ from $R$ or simply find $k$ linearly independent points $K$ which are not in the span of $d-1$ or less points of $P$.
\item
Project $\mathcal X$ from the $(k-1)$-dimensional subspace spanned by the points of $K$ onto a quantum set of lines $\mathcal X'$ of PG$(n-k-1,p)$. 
\item
Calculate the set $R'$ of points of PG$(n-k-1,p)$ which are not in the span of $d-1$ or less points of $P'$, the points incident with lines of $\mathcal X'$. 
\item
Make a graph $\Gamma$ whose vertices are the points in $R'$ and where $u,v$ are joined by an edge if and only if the subspace spanned by $u$ and $v$ and any $d-1$ points of $P'$ has (projective) dimension $d$, i.e. these $d+1$ points are linearly independent.
\item
Find a large, preferably the largest, clique $C$ in the graph $\Gamma$. 
\item
Let $T$ be the subset of ${\mathbb F}_p^{n-k}$ which contains the zero vector and any vector which spans a one-dimensional subspace which is a projective point in $C$ and let $S$ be the abelian subgroup obtained from $\mathcal X'$. 
\item
Then $Q(S,T)$ is a $(\!(n,p^k|T|,d)\!)_p$ code.
\end{itemize} 

This generalises the method set out in \cite{YCO2007} which is a combinatorial interpretation of this method in the case $k=0$ and $p=2$. Note that if $k=0$ the graph $\Gamma$ will often be so large that finding a large clique $C$ will be hard. The advantage here is that we can choose $k$ large enough, so that the graph $\Gamma$, which has less than $p^{n-k}$ vertices, is small enough to allow clique finding algorithms to be implemented. The example in Section~\ref{secgraph} indicates that another trick is to restrict the vertices of $\Gamma$ to a well chosen subspace $\pi$, which has a small intersection with $R$ (or $R'$ if we choose $k>0$). This again reduces the size of the graph $\Gamma$ so that clique finding algorithms can be implemented.

\vspace{1cm}

   Simeon Ball\\
   Departament de Matem\`atiques, \\
Universitat Polit\`ecnica de Catalunya, \\
M\`odul C3, Campus Nord,\\
Carrer Jordi Girona 1-3,\\
08034 Barcelona, Spain \\
   {\tt simeon@ma4.upc.edu} \\

Pablo Puig\\
Facultat de Matem\`atiques, \\
Universitat Polit\`ecnica de Catalunya, \\
Carrer de Pau Gargallo, 14, \\
08028 Barcelona, Spain\\
  {\tt pablo.puig@estudiantat.upc.edu} \\

 \end{document}